\documentclass[prl,aps,twocolumn,superscriptaddress,showpacs,showkeys]{revtex4-1}
\usepackage{amsmath,amsthm,amssymb,graphicx,subfigure,xcolor}
\usepackage[unicode=true]{hyperref}
\hypersetup{
     colorlinks=true,       		% false: boxed links; true: colored links
     linkcolor=blue,          	% color of internal links
     citecolor=red,            % color of links to bibliography
     urlcolor=magenta,           	% color of external links
 }

\newtheorem{theorem}{Theorem}
\newtheorem*{theorem*}{Theorem}

\newtheorem*{corollary*}{Corollary}

\newtheorem*{lemma*}{Lemma}

\newtheorem*{proposition*}{Proposition}
\theoremstyle{definition}

\newtheorem*{definition*}{Definition}
\theoremstyle{remark}
\newtheorem{remark}{Remark}
\newtheorem*{remark*}{Remark}
\begin{document}

\title{Generalized Hardy's Paradox}

\author{Shu-Han Jiang}
\affiliation{Theoretical Physics Division, Chern Institute of Mathematics, Nankai University, Tianjin 300071, People's Republic of China}
\affiliation{School of Physics, Nankai University, Tianjin 300071, People's Republic of China}

\author{Zhen-Peng Xu}
\affiliation{Theoretical Physics Division, Chern Institute of Mathematics, Nankai University, Tianjin 300071, People's Republic of China}
\affiliation{Departamento de F\'{\i}sica Aplicada II, Universidad de  Sevilla, E-41012 Sevilla, Spain}

\author{Hong-Yi Su}
\email{hysu@gscaep.ac.cn}
\affiliation{Graduate School of China Academy of Engineering Physics, Beijing 100193, People's Republic of China}

\author{Arun Kumar Pati}
\email{akpati@hri.res.in}
 \affiliation{Quantum Information and Computation Group, Harish-Chandra Research Institute, Chhatnag Road, Jhunsi, Allahabad 211 019, India}

\author{Jing-Ling Chen}
\email{chenjl@nankai.edu.cn}
\affiliation{Theoretical Physics Division, Chern Institute of Mathematics, Nankai University, Tianjin 300071, People's Republic of China}
\affiliation{Centre for Quantum Technologies, National University of Singapore, 3 Science Drive 2, Singapore 117543}

\date{\today}
\begin{abstract}
Here we present the most general framework for $n$-particle Hardy's paradoxes, which include Hardy's original one and Cereceda's extension as special cases. Remarkably, for any $n\ge 3$ we demonstrate that there always exist generalized paradoxes (with the success probability as high as $1/2^{n-1}$) that are stronger than the previous ones in showing the conflict of quantum mechanics with local realism. An experimental proposal to observe the stronger paradox is also presented for the case of three qubits. Furthermore, from these paradoxes we can construct the most general Hardy's inequalities, which enable us to detect Bell's nonlocality for more quantum states.
\end{abstract}


\maketitle

\emph{Introduction.}---Hardy's paradox is an important all-versus-nothing (AVN) proof of Bell's nonlocality, a peculiar phenomenon that has its roots deep in the famous debate raised by Einstein, Podolsky and Rosen (EPR) in 1935~\cite{EPR}. Hardy's original proof~\cite{Hardy92,Hardy93}, for two particles, has been considered as ``the simplest form of Bell's theorem'' and ``one of the strangest and most beautiful gems yet to be found in the extraordinary soil of quantum mechanics''~\cite{Mermin94}. To date, a number of experiments has been carried out to confirm the paradox in two-particle systems~\cite{Hardy-Exp-1,Hardy-Exp-2,Hardy-Exp-3,Hardy-Exp-4,Hardy-Exp-5,Hardy-Exp-6,Hardy-Exp-7,Hardy-Exp-8,Hardy-Exp-9}; theoretically, Hardy's paradox has been generalized from the two-qubit to a multi-qubit family \cite{Cereceda2004}.
The two-particle Hardy's paradox can be stated in an inspiring way as follows~\cite{Chen2003}: In any local theory, if the events $A_2<B_1$, $B_1<A_1$, and $A_1<B_2$ never happen, then naturally the event $A_2 < B_2$ must never happen. According to quantum theory, however, there exist two-particle entangled states and local projective measurements that break down these local conditions; that is, in terms of probabilities,
%\begin{subequations}
% \label{E1}
 \begin{align*}
 &P(A_2 < B_1) = %\label{E1a}\\
 P(B_1 < A_1) = %\label{E1b}\\
 P(A_1 < B_2) = 0,\\ %\label{E1c}\\
 &~~{\rm and}~~P(A_2 < B_2) > 0,  %\label{E1d}
\end{align*}
%\end{subequations}
where the last condition evidently conflicts with the prediction of local theory, leading to a paradox. In \cite{Cereceda2004} the author showed that for the $n$-qubit Greenberger-Horne-Zeilinger (GHZ) state the maximal success probability (i.e., the last condition above) can reach $[1+\cos\frac{\pi}{n-1}]/2^{n}$.

Moreover, a quantum paradox can be naturally transformed to a corresponding Bell's inequality. For instance, the paradox mentioned above can be associated to the following Hardy's inequality $P(A_2 < B_2)-P(A_2 < B_1)-P(B_1 < A_1)-P(A_1 < B_2) \leq 0 $,
%\begin{eqnarray}\label{gill}
%&&P(A_2 < B_2)-P(A_2 < B_1)\nonumber\\
%&&\;\;\;-P(B_1 < A_1)-P(A_1 < B_2) \leq 0.
% \end{eqnarray}
which is  equivalent to Zohren and Gill's
version~\cite{Gill2008}  of the Collins-Gisin-Linden-Massar-Popescu inequalities (i.e., tight Bell's inequalities for two arbitrary $d$-dimensional systems, and the inequality becomes the CHSH  inequality for $d=2$)~\cite{CGLMP}. See also~\cite{saha15} for a connection between Hardy's inequality and Wigner's argument.

Demonstrating the conflict between quantum mechanics and local theories has had a long history ever since the EPR paper. It has brought out many important contributions to both physical foundations and applications, particularly introducing the concept of \emph{entanglement}, viewed as ``the characteristic trait of quantum mechanics'' that distinguishes quantum theory from classical theory~\cite{Schrodinger35}. Among many others, the most important breakthrough was due to Bell who put the debate of the conflict on firm, physical ground in a statistical manner~\cite{Bell}, and it has been regarded as ``the most profound discovery of science''~\cite{Stapp}.
%In 1964, Bell made an important step forward by discovering that entangled states of two spin-$1/2$ particles may violate a certain inequality, which holds for any local-hidden-variable (LHV) model. Violation of Bell's inequality by quantum states implies Bell's nonlocality, indicating no LHV models can completely reproduce all quantum predictions by performing measurements on the quantum states. Thus Bell's inequality has provided a feasible tool for the first time to experimentally test local realism. For instance,
The Clause-Horne-Shimony-Holt (CHSH) inequality~\cite{CHSH}, serving as a revised version of Bell's original one, has been adopted to reveal nonlocality in various experiments, ranging from Aspect's experiment~\cite{Aspect81} in 1981 to some very recent loophole-free Bell-experiment tests~\cite{Hensen,Giustina,Shalm}. On the other hand, differing from the statistical violation of inequalities, the AVN proof of nonlocality allows to demonstrate contradiction in an elegant, logic paradox, such that its experimental practice will be, in principle, simplified to a single-run operation. Among various AVN proofs, the GHZ paradox~\cite{GHZ89} has been carried out experimentally based on entangled photons~\cite{Pan2000}. In spite of that, it applies to three-particle systems~\cite{GHZ89} or more~\cite{HorodeckiRMP,BrunnerRMP}, but has so far defied any two-particle formulation.

Hardy's paradox, with post-selections taken into consideration, therefore stands out among the others, since (i) it applies to the two-party scenario; (ii) it can be generalized to multi-party and high-dimensional scenarios~\cite{Cereceda2004} (hereafter we would like to call Cereceda's version of $n$-qubit Hardy's paradox/inequality as the standard Hardy's paradox/inequality, to distinguish them from the most general ones that we shall present in this paper); and (iii) inequalities constructed based on it allow to detect more entangled states and provide a key element to prove Gisin's theorem~\cite{Choudhary2010,Yu2012} --- which states that any entangled pure state violates Bell's inequality~\cite{Gisin1991}. The GHZ paradox does not share most of these merits (see also the Mermin-Ardehali-Belinskii-Klyshko  inequality~\cite{Mermin1990,Ardehali,BK}, which was also a kind of generalization of CHSH inequality to $n$ qubits, but was not violated by all pure entangled states, even not by all the generalized GHZ states).

In this Letter,
%we advance the study of Hardy's paradox and Hardy's inequality for multiparticles.
we first present a family of generalized Hardy's paradoxes for $n$ qubits and show that the standard Hardy's paradox is a special case of the family, and that for any $n\geq 3$ one can always have a stronger quantum paradox in comparison to the standard one. Then, we present a family of generalized Hardy's inequalities based on the generalized paradoxes and show that, similar to the paradox, the standard Hardy's inequality is a special case of the family of generalized Hardy's, that some of the generalized Hardy's inequalities are tight based on the numerical computation, and that the generalized Hardy's inequalities for $n\ge 4$ are  stronger than the standard Hardy's inequality based on the visibility criterion.
%(iii) one can use the general Hardy's inequalities to prove the $n$-qubit Gisin's theorem.
An experimental proposal to observe the stronger quantum paradoxes in a three-qubit system is also presented.

%Our results not only enhance the study of Bell's nonlocality by presenting the most general framework for generalized Hardy's paradox and inequality, but also provide a feasible proposal to experimentally observe the stronger paradox.

\emph{Generalized Hardy's Paradox.}---For simplicity, we shall use the notations in \cite{Yu2012} to formalize the generalized $n$-qubit Hardy's paradox.
Consider a system composed of $n$ qubits that are labeled with the index set $I_n=\{1, 2,..., n\}$. For the $k$-th qubit, we choose two observables $\{a_k, b_k\}$ that
take binary values $\{0, 1\}$ in the local realistic model. Let us denote $a_\alpha=\prod_{k \in\alpha} a_k$ and $\overline{b}_\alpha=\prod_{k \in\alpha} \overline{b}_k$ with $\overline{b}_k=1-b_k$ for an arbitrary subset $\alpha\subseteq I_n$, $\overline{k}=I_n/k$ for arbitrary $k \in I_n$ and $\overline{\alpha}=I_n/\alpha$.
Moreover, we denote $|\alpha|$ as the size of the subset $\alpha$, and abbreviate the probability $p(x=1,y=1, \ldots)$ as $p(xy\ldots)$.

We now present the following theorem:
\begin{theorem}
  For any given sizes $|\alpha|$ and $|\beta|$ $(2 \leq |\alpha| \leq n, 1\leq |\beta|\leq |\alpha|)$ satisfying the constraint $|\alpha|+|\beta| \leq n+1$,
%\begin{eqnarray}\label{alphabeta}
%|\alpha|+|\beta| \leq n+1,
%\end{eqnarray}
then in the LHV model, the following zero-probability conditions
%\begin{subequations}
% \label{E2}
% \begin{align}
% p(b_{\alpha}a_{\overline{\alpha}}) &= 0, \label{E2a}\\
% p(\overline{b}_{\beta}a_{\overline{\beta}}) &= 0, \;\;\;\; \forall \alpha,\beta \in I_n\label{E2b}
%\end{align}
%\end{subequations}
 \begin{align*}
 p(b_{\alpha}a_{\overline{\alpha}}) =  p(\overline{b}_{\beta}a_{\overline{\beta}}) = 0, \;\;\;\; \forall \alpha,\beta \in I_n %\label{E2}
\end{align*}
must lead to the following zero-probability condition
\begin{eqnarray*}%\label{E3}
&p(a_{I_n}) = 0.
\end{eqnarray*}
\end{theorem}

%The proof of the theorem is provided in Appendix A.

\begin{proof} Note that the above equations are all linear for the LHV model which is a convex polytope, whose extreme points are the deterministic LHV model. Thus, we only need to prove this theorem for the deterministic LHV model, that is,
%\begin{subequations}
% \label{E22}
% \begin{align}
% b_{\alpha}a_{\overline{\alpha}} &= 0, \;\;\;\; (\forall \alpha \in I_n)\label{E2a2}\\
% \overline{b}_{\beta}a_{\overline{\beta}} &= 0, \;\;\;\; (\forall \beta \in I_n)\label{E2b2}
%\end{align}
%\end{subequations}
\begin{align*}
 b_{\alpha}a_{\overline{\alpha}} =
 \overline{b}_{\beta}a_{\overline{\beta}} = 0, \;\;\;\; \forall \alpha,\beta \in I_n %\label{E22}
\end{align*}
must lead to the following zero-probability condition
\begin{eqnarray*}%\label{E32}
&a_{I_n} = 0.
\end{eqnarray*}
We shall prove it by \emph{reductio ad absurdum}. Suppose $a_{I_n} \neq 0$, then
%from Eqs. (\ref{E22})
in the deterministic LHV model one directly obtains
%\begin{subequations}
% \label{E4}
% \begin{align}
%  b_{\alpha}=0, \;\;\;\; (\forall \alpha \in I_n)\label{E4a}\\
%  \overline{b}_{\beta}=0, \;\;\;\; (\forall \beta \in I_n)\label{E4b}
%\end{align}
%\end{subequations}
\begin{align*}
  b_{\alpha}=
  \overline{b}_{\beta}=0, \;\;\;\; \forall \alpha,\beta \in I_n,  %\label{E4}
\end{align*}
which implies at least one of $|\beta|$ observables $b_k$'s arbitrarily chosen from the set $\mathcal{B}=\{b_1, b_2, ..., b_n\}$ must take the value ``1'' ---
%Let us choose arbitrary  from the set $\mathcal{B}=\{b_1, b_2, ..., b_n\}$, then Eq. (\ref{E4}) implies at least one of them takes the value ``1'',
namely, in the set $\mathcal{B}$ we have at least $n-(|\beta|-1)$ observables equal to 1 --- and which, similarly, implies at least one of $|\alpha|$ observables $b_k$'s arbitrarily chosen from the set $\mathcal{B}$ must take the value ``0''
%let us choose arbitrary $|\alpha|$ observables $b_k$'s from the set $\mathcal{B}$, then Eq. (\ref{E4}) implies at least one of them takes the value ``0'',
--- namely, in the set $\mathcal{B}$ we have at least $n-(|\alpha|-1)$ observables equal to 0.
Hence, at most $(|\alpha|-1)$ observables $b_k$'s equal 1. This yields $(|\alpha|-1)\geq n-(|\beta|-1)$, i.e., $|\alpha|+|\beta| \ge n+2$, in contradiction to the constraint $|\alpha|+|\beta| \leq n+1$.
\end{proof}

For the sake of convenience, we label the generalized paradox as $[n;|\alpha|,|\beta|]$-scenario. It can be verified directly that the standard Hardy's paradox is the $[n;n,1]$-scenario by taking $|\alpha|=n, |\beta|=1$. Nevertheless, quantum mechanics gives different prediction that the success probability $p(a_{I_n})$ can be non-zero, thus resulting in a generalized Hardy's paradox, stated as:
\begin{theorem}
  For the generalized GHZ state, by choosing appropriate quantum projective measurements on $n$ qubits, the success probability $p(a_{I_n})$ is always greater than zero, and for any $n\geq 3$ we can always have a stronger quantum paradox in comparison to the standard Hardy's paradox.
\end{theorem}
\begin{proof} Quantum mechanically, let us consider the generalized GHZ state
\begin{eqnarray*}%\label{GGHZ}
& |\Psi\rangle_{\rm gGHZ}=h_{0}|0\ldots 0\rangle + h_{1} |1\ldots 1\rangle,
\end{eqnarray*}
with $h_0=|h_0|\geq 0$, $h_1=|h_1|e^{i \theta_{h}}$ (The usual GHZ state corresponds to $|h_0|=|h_1|=1/\sqrt{2}, \;\theta_{h}=0 $). We always assume the measurements $a_{i}$'s, $b_{i}$'s, and $\bar{b}_{i}$'s for the $n$ observers are in the direction $a_{0}|0\rangle + a_{1}e^{i\theta_{a}} |1\rangle$, $b_{0}|0\rangle + b_{1}e^{i\theta_{b}} |1\rangle$ and $b_{1}|0\rangle + b_{0}e^{i(\theta_{b}+\pi)} |1\rangle$ respectively, by direct calculation, we then obtain
\begin{align*}
  p(b_{\alpha}a_{\bar{\alpha}}) &=   \biggr|b_{0}^{|\alpha|}a_{0}^{n-|\alpha|}|h_{0}|+ b_{1}^{|\alpha|}a_{1}^{n-|\alpha|}|h_{1}| e^{i\vartheta}\biggr|^{2},\\
  p(\bar{b}_{\beta}a_{\bar{\beta}}) &=  \biggr |b_{1}^{|\beta|}a_{0}^{n-|\beta|}|h_{0}|
  + b_{0}^{|\beta|}a_{1}^{n-|\beta|}|h_{1}| e^{i\vartheta'}\biggr|^{2},\\
    p(a_{N}) &=  \biggr |a_{0}^{n}|h_{0}| + a_{1}^{n}|h_{1}| e^{i[n\theta_{a}-\theta_{h}]}\biggr|^{2},
\end{align*}
%\begin{align*}
%  &p(b_{\alpha}a_{\bar{\alpha}}) =   \biggr|b_{0}^{|\alpha|}a_{0}^{n-|\alpha|}|h_{0}| \\
%   &~~~~~~~~~~~~~~~~~~~~+ b_{1}^{|\alpha|}a_{1}^{n-|\alpha|}|h_{1}| e^{i[(n-|\alpha|)\theta_{a}+|\alpha|\theta_{b}-\theta_{h}]}\biggr|^{2},\\
%  &p(\bar{b}_{\beta}a_{\bar{\beta}}) =  \biggr |b_{1}^{|\beta|}a_{0}^{n-|\beta|}|h_{0}| \\
%  &~~~~~~~~~~~~~~+ b_{0}^{|\beta|}a_{1}^{n-|\beta|}|h_{1}| e^{i[(n-|\beta|)\theta_{a}+|\beta| \theta_{b}-\theta_{h}+|\beta|\pi]}\biggr|^{2},\\
%   & p(a_{N}) =  \biggr |a_{0}^{n}|h_{0}| + a_{1}^{n}|h_{1}| e^{i[n\theta_{a}-\theta_{h}]}\biggr|^{2},
%\end{align*}
respectively, where $\vartheta=(n-|\alpha|)\theta_{a}+|\alpha|\theta_{b}-\theta_{h}$ and $\vartheta'=(n-|\beta|)\theta_{a}+|\beta| \theta_{b}-\theta_{h}+|\beta|\pi$.

Let $p(b_{\alpha}a_{\bar{\alpha}}) = p(\bar{b}_{\beta}a_{\bar{\beta}}) = 0$, we have equations of angles
\begin{align*}
    &(n-|\alpha|)\theta_{a}+|\alpha|\theta_{b}-\theta_{h}=(2 m_1 +1)\pi,\\
    &(n-|\beta|)\theta_{a}+|\beta| \theta_{b}-\theta_{h}+|\beta|\pi=(2 m_2 +1)\pi,
\end{align*}
with $m_1, m_2=0, 1, 2,...$, and of norms
\begin{align*}%\label{key}
    b_0^{|\alpha|}a_0^{n-|\alpha|} |h_0|=b_1^{|\alpha|}a_1^{n-|\alpha|} |h_1|,\\
    b_1^{|\beta|}a_0^{n-|\beta|} |h_0|=b_0^{|\beta|}a_1^{n-|\beta|} |h_1|.
        \end{align*}
% For $n=2$, the successful probability has been proved to be greater than 0 in the original two-qubit Hardy's paradox. In the following, we only need to prove the situation for %$n\geq 3$.
The following arguments are split into two cases:\\
%\begin{enumerate}
%  \item
 \emph{Case 1:} $|\beta|<|\alpha|$: we let $m_1=m_2=0$, then we have $\theta_{b}=\frac{|\beta|\pi}{|\alpha|-|\beta|}+\theta_{a}$, $n \theta_{a}-\theta_{h}=(1-\frac{|\alpha||\beta|}{|\alpha|-|\beta|})\pi$, and  $\frac{a_0}{a_1}=(\frac{b_0}{b_1})^{\frac{|\alpha|+|\beta|}{|\alpha|-|\beta|}}=\gamma^{\frac{|\alpha|+|\beta|}{(|\alpha|+|\beta|)n-2|\alpha||\beta|}}$, with $\gamma=\frac{|h1|}{|h0|}$, and so the success probability equals
% (denote $P_{\rm suc}=p(a_{I_n})$)
   \begin{align*}
    &p(a_{I_n})   =\frac{\gamma^2|\kappa_0-\kappa_1|^2}
    {(1+\gamma^{2})(1+\kappa_2)^{n}}>0,~~\kappa_0=e^{i\frac{|\alpha||\beta|}{|\alpha|-|\beta|}\pi},\\
    &~~\kappa_1=\gamma^{\frac{2|\alpha||\beta|}{n(|\alpha|+|\beta|)-2|\alpha||\beta|}},~~\kappa_2=\gamma^{\frac{2(|\alpha|+|\beta|)}{n(|\alpha|+|\beta|)-2|\alpha||\beta|}}.
   \end{align*}
At $\gamma=1$, on the other hand, the success probability equals
   \begin{eqnarray*}%\label{sucgen}
   p(a_{I_n}) &=&\frac{1}{2^{n}}\biggr[1-\cos\biggr(\frac{|\alpha||\beta|\pi}{|\alpha|-|\beta|}\biggr)\biggr].
   \end{eqnarray*}
Note that  $ p(a_{I_n})$ is strictly smaller than $\frac{1}{2^{n-1}}$ because $\frac{|\alpha||\beta|}{|\alpha|-|\beta|}$ cannot be odd. For the standard Hardy's paradox, i.e., $\alpha=n, \beta=1$, it reduces to the result in~\cite{Cereceda2004} as
  \begin{eqnarray}\label{sucorig}
  P^{\rm S}_{n}\equiv p(a_{I_n}) &=&\frac{1}{2^{n}}\biggr[1+\cos\biggr(\frac{\pi}{n-1}\biggr)\biggr],
   \end{eqnarray}
where $P^{\rm S}_{n}$ represents the success probability for the standard Hardy's paradox for the $n$-qubit GHZ state.\\
  %\item
  \emph{Case 2:} $|\beta|=|\alpha|$: we let $m_1=0, m_2=\frac{|\beta|}{2}$ (here $|\alpha|$ and $|\beta|$ must be even). Note that we have in this case an independent $\theta_h$, then we further let $n\theta_{a}-\theta_{h}=0$ and  $b_0=b_1=1/\sqrt{2}$, $\frac{a_0}{a_1}=\gamma^{\frac{1}{n-|\alpha|}}$, with $\gamma=\frac{|h1|}{|h0|}$, and the success probability equals
 %   \begin{equation}
%    \begin{split}
%    P_s =&|a_0^n |h_0|+a_1^n |h_1||^2\\
%    =&\left(\gamma ^{\frac{|\alpha| }{n-|\alpha|}}+1\right)^2 \left(\gamma ^{\frac{2}{n-|\alpha| }}+1\right)^{-n} \left(\gamma ^{-2}+1\right)^{-1} >0.
%    \end{split}
%    \end{equation}
     \begin{align*}%\label{sucgen2}
   p(a_{I_n})&=\frac{\gamma^2(1+\kappa'_1)^2}{(1+\gamma^2)(1+\kappa'_2)^{n} }>0,\\
   \kappa'_1&=\gamma ^{\frac{|\alpha| }{n-|\alpha|}},~~~\kappa'_2=\gamma ^{\frac{2}{n-|\alpha| }}.
   \end{align*}
The success probability at $\gamma=1$ equals
    \begin{equation*}
  P^{\rm G}_{n}\equiv p(a_{I_n}) =\frac{1}{2^{n-1}},
    \end{equation*}
where $P^{\rm G}_{n}$ represents the success probability for the generalized Hardy's paradox for the $n$-qubit GHZ state.

%\end{enumerate}
Combining the above two cases, the theorem is proved as was claimed.
\end{proof}

\begin{remark}
  As an example, given the GHZ state $|\Psi\rangle_{\rm GHZ}=(|00\cdots0\rangle + |11\cdots1\rangle)/\sqrt{2}$ of $n$ qubits, Cereceda~\cite{Cereceda2004} found that the maximal success probability for the standard Hardy's paradox is Eq. (\ref{sucorig})
But, by choosing $|\beta|=|\alpha|$ in the generalized Hardy's paradox, for any $n\ge 3$ we can have a greater success probability (see also Fig.~\ref{fig1}):
\begin{equation}
  P^{\rm G}_{n}\equiv p(a_{I_n}) =\frac{1}{2^{n-1}} > P^{\rm S}_{n}.
    \end{equation}
Indeed for GHZ states with $n\geq 3$, $|\alpha|=|\beta|= {\rm even\; number}$ is the best choice for generalized Hardy's paradoxes~\cite{supp}.
\end{remark}

\begin{remark}
The $[n;|\alpha|=2,|\beta|=1]$-scenario resembles the paradox presented in~\cite{chen2014}, but the former concerns the Bell scenario, while the latter discusses the genuine multipartite nonlocality, which is a subset of the Bell nonlocality; the $[n;2<|\alpha|<n,|\beta|=1]$-scenario is related to the paradox presented in~\cite{wang2016}, which discussed hierarchy of multipartite nonlocality. It is thus of great interest to further investigate possible connections of the results in \cite{chen2014} and \cite{wang2016} with the structure of \textbf{Theorem 1}.
\end{remark}

\begin{remark}
  For the paradox of $[n;|\alpha|,|\beta|]$-scenario, one can have its corresponding generalized Hardy's inequality as
\begin{eqnarray}%\label{eq:genhardy2}
&\mathcal{I}[n;|\alpha|,|\beta|;x,y]= F(n;\alpha,\beta; x, y) \; p(a_{I_n})\nonumber\\
 &\;\;\;\;\;\;\;\;\;\;\;\;- x \sum_{\alpha} p(b_{\alpha} a_{\bar{\alpha}}) - y \sum_{\beta } p(\bar{b}_{\beta} a_{\bar{\beta}}) \leq 0,
\end{eqnarray}
with  $x>0, y>0$. Usually for convenience, one can choose $x, y$ as positive integers, and to make the inequality meaningful (i.e., it can be possibly violated by quantum states), one needs to require $F(\alpha,\beta; x, y)>0$. By directly computation, one can determine
\begin{equation*}
F(n;\alpha,\beta; x, y) =\min_{0\le m \le n} \left(x\binom{m}{|\alpha|} + y \binom{n-m}{|\beta|}\right),
\end{equation*}
which is the largest integer that the inequality still holds, and $\binom{m}{k}=\frac{m!}{k!(m-k)!}$ is the binomial coefficient.   For $x=y=1, |\beta|=1$, one has the coefficient as $F(\alpha,\beta; 1, 1)=n-|\alpha|+1$. For $x=y=1$, the family of the generalized Hardy's inequalities is particularly interesting, one may have that: \\
(i) The standard $n$-qubit Hardy's inequality corresponds to $\mathcal{I}[n;|\alpha|=n,|\beta|=1;x=1,y=1]$, which is a family of tight Bell's inequalities; the 22nd Sliwa's inequality~\cite{Sliwa2003} corresponds to $\mathcal{I}[n=3;|\alpha|=2,|\beta|=1;x=1,y=1]$, which is a tight Bell's inequality; also, numerical computation shows that the family of $n$-qubit Bell's inequalities  $\mathcal{I}[n;|\alpha|,|\beta|=1;x=1,y=1]$ is tight~\cite{othertight}; \\
(ii) Based on the visibility criterion, for $n\ge 4$, the generalized Hardy's inequalities can resist more white-noise than the standard Hardy's inequality. For a given $n$-qubit entangled state $\rho$, we can mixed it with the white noise $I_{\rm noise}=\openone^{\otimes n}/2^n$, the resultant density matrix is given by $\rho^V =V \rho+ (1-V) I_{\rm noise}$.
%  \begin{eqnarray}\label{density}
%   \rho^V &=&V \rho+ (1-V) I_{\rm noise}.
%     \end{eqnarray}
 Resistance to noise can be measured through the threshold visibility $V_{\rm thr}$, below which Bell's inequality cannot be violated. A lower threshold visibility means that quantum state can tolerate a greater amount of noise. Let us consider $\rho$ as the $n$-qubit GHZ state. In Table~\ref{vis}, we compare the threshold visibility of the generalized Hardy's inequalities and that of the standard Hardy's inequality.  We find that for $n\ge 4$, the generalized Hardy's inequalities can provide lower visibilities than the standard one.
\end{remark}

\begin{figure}[t]
\includegraphics[width=65mm]{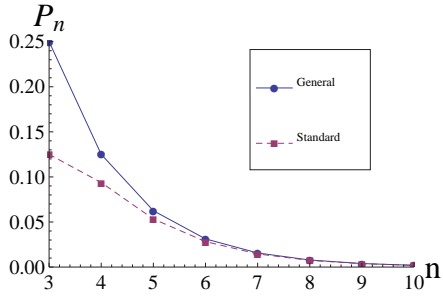}\\
\caption{(color online). The success probability $P_{n}$ versus particle number $n$ $(n\ge 3)$. The blue points correspond to $P^{\rm G}_{n}$ in the generalized Hardy's paradox (with $|\alpha|=|\beta|={\rm even\; number}$), and  the purple points correspond to $P^{\rm S}_{n}$ in the standard Hardy's paradox (with $|\alpha|=n, |\beta|=1$). The relation $P^{\rm G}_{n}> P^{\rm S}_{n}$
implies that there are always a stronger paradox in comparison to the standard Hardy's paradox for $n\ge 3$.
Especially, in the $[3;2,2]$-scenario we have $P^{\rm G}_{3}=1/4$, which is twice of $P^{\rm S}_{3}=1/8$. Thus, it is feasible to observe the stronger quantum paradox in a three-qubit system.
  }\label{fig1}
\end{figure}

\begin{table*}[htbp]
\caption{\label{vis}Numerical results of threshold visibility $V_{\rm thr}[n;|\alpha|,|\beta|;1,1]$ for violations of inequality $\mathcal{I}[n;|\alpha|,|\beta|;1,1]\le 0$ by the $n$-qubit GHZ states. The boxed number represents the lowest visibility for each $n$. For $n\ge 4$, the generalized Hardy's inequalities can provide lower visibility than the standard one (which corresponds to $|\alpha|=n, |\beta|=1$). For $|\alpha|=|\beta|=q$,  ( $q$ is even, $2q\le n+1$), one may have the analytical expression $V_{\rm thr}[n;q,q;1,1]=\frac{2\binom{n}{q}-\binom{[n/2]}{q}-\binom{n-[n/2]}{q}}{2\binom{n}{q}+\binom{[n/2]}{q}+\binom{n-[n/2]}{q}}$. For $q=2$, we have $V_{\rm thr}[n;2,2;1,1]=	 \frac{3[\frac{n+1}{2}]-1}{5[\frac{n+1}{2}]-3}$. It can be proved that for the case of $x=y=1$, for GHZ states with $n \geq 5$, the relation $|\alpha|=|\beta|=2$ is the best choice for generalized Hardy's inequality~\cite{supp}.}
\begin{ruledtabular}
\begin{tabular}{lccccccccc}
  $n$&$3$&$4$&$5$&$6$&$7$&$8$&$9$&$10$\\
	\hline
			$|\alpha|=n,|\beta|=1 $ & \fbox{$0.681250$}&$0.707107$&$0.737431$&$0.764501$&$0.787467$&$0.806795$&$0.823130$&$0.837049$\\
			 $|\alpha|=2,|\beta|=1$&$0.682242$&$0.703526$&$0.730699$&$0.755929$&$0.777878$&$0.796691$&$0.812819$&$0.826718$\\
			 $|\alpha|=n-1,|\beta|=1$&$0.682242$&$\fbox{0.671442}$&$0.702481$&$0.734966$&$0.763073$&$0.786584$&$0.806221$&$0.822742$\\
 $|\alpha|=2,|\beta|=2$&$0.714286$&$0.714286$&$\fbox{0.666667}$&$\fbox{0.666667}$&$\fbox{0.647059}$&$\fbox{0.647059}$&
 $\fbox{0.636364}$&$\fbox{0.636364}$
\end{tabular}
\end{ruledtabular}

\end{table*}

\emph{Experimental proposal to observe the stronger paradox with three qubits.}---A number of experimental tests of the two-qubit Hardy¡¯s paradox have
been carried out since 1993~\cite{Hardy-Exp-1,Hardy-Exp-2,Hardy-Exp-3,Hardy-Exp-4,Hardy-Exp-5,Hardy-Exp-6,Hardy-Exp-7,Hardy-Exp-8,Hardy-Exp-9}. The maximal success probability for two-qubit Hardy's paradox is $(5\sqrt{5}-11)/2\simeq9\%$, which does not occur for the maximally entangled state~\cite{Hardy93}\cite{Cereceda2004}. For the three-qubit standard Hardy's paradox, the success probability is given by $P^{\rm S}_{3}=1/8=0.125$, which occurs for the GHZ state. To our knowledge, such an experiment has not yet been demonstrated. The higher the success probability, the more friendly the experimental observation. Here, we present an experimental proposal to observe stronger paradox in the $[n=3;|\alpha|=2,|\beta|=2]$-scenario, whose success probability is $P^{\rm G}_{3}=1/4=0.25$. In the experiment, the resource is prepared as the three-qubit GHZ state $|\Psi\rangle_{\rm GHZ}=(|000\rangle + |111\rangle)/\sqrt{2}$,
%\begin{eqnarray}\label{GHZ2}
%& |\Psi\rangle_{\rm GHZ}=\frac{1}{\sqrt{2}}(|000\rangle + |111\rangle),
%\end{eqnarray}
and three qubits are sent to three observers Alice, Bob and Charlie separately (see Fig.~\ref{fig2}). Quantum mechanically, the three observers will all perform the same measurements in $\hat{x}$- and $\hat{y}$-direction respectively, i.e.,
\begin{eqnarray*}%\label{measure}
\hat{a}_1=\hat{a}_2=\hat{a}_3= |+x\rangle \langle+x|,~~
\hat{b}_1=\hat{b}_2=\hat{b}_3= |+y\rangle \langle+y|,
\end{eqnarray*}
with $\overline{\hat{b}}_j=\openone-\hat{b}_j=|-y\rangle \langle-y|$, $(j=1,2,3)$, $\openone$ is the $2\times 2$ unit matrix, and $|+x\rangle =\frac{1}{\sqrt{2}}(|0\rangle +|1\rangle)$,
$|\pm y\rangle =\frac{1}{\sqrt{2}}(|0\rangle \pm i|1\rangle)$.

Firstly one needs to experimentally verify the zero-probability conditions, i.e.,
\begin{equation}\label{check1}
\begin{split}
&p(\hat{b}_1 \hat{b}_2 \hat{a}_3)= p(\hat{b}_1 \hat{a}_2 \hat{b}_3)=p(\hat{a}_1 \hat{b}_2 \hat{b}_3)\\
&~~=p(\overline{\hat{b}}_1 \overline{\hat{b}}_2 \hat{a}_3)= p(\overline{\hat{b}}_1 \hat{a}_2 \overline{\hat{b}}_3)=p(\hat{a}_1 \overline{\hat{b}}_2 \overline{\hat{b}}_3)=0,
\end{split}
\end{equation}
with $p(\hat{b}_1 \hat{b}_2 \hat{a}_3)= {\rm tr}[\rho (\hat{b}_1 \otimes\hat{b}_2\otimes \hat{a}_3)]$, etc, and $\rho$ stands for the GHZ state.
%\begin{eqnarray}\label{check1}
%&p(\hat{b}_1 \hat{b}_2 \hat{a}_3)= {\rm tr}[\rho (\hat{b}_1 \otimes\hat{b}_2\otimes \hat{a}_3)]=0, \nonumber\\
%&p(\hat{b}_1 \hat{a}_2 \hat{b}_3)= {\rm tr}[\rho (\hat{b}_1 \otimes\hat{a}_2\otimes \hat{b}_3)]=0, \nonumber\\
%&p(\hat{a}_1 \hat{b}_2 \hat{b}_3)= {\rm tr}[\rho (\hat{a}_1 \otimes\hat{b}_2\otimes \hat{b}_3)]=0,
%\end{eqnarray}
%and
%\begin{eqnarray}\label{check2}
%&p(\overline{\hat{b}}_1 \overline{\hat{b}}_2 \hat{a}_3)= {\rm tr}[\rho (\overline{\hat{b}}_1 \otimes\overline{\hat{b}}_2\otimes \hat{a}_3)]=0, \nonumber\\
%&p(\overline{\hat{b}}_1 \hat{a}_2 \overline{\hat{b}}_3)= {\rm tr}[\rho (\overline{\hat{b}}_1 \otimes \hat{a}_2 \otimes\overline{\hat{b}}_3]=0, \nonumber\\
%&p(\hat{a}_1 \overline{\hat{b}}_2 \hat{\overline{b}}_3)= {\rm tr}[\rho (\hat{a}_1 \otimes \overline{\hat{b}}_2 \otimes \hat{\overline{b}}_3)]=0.
%\end{eqnarray}
  Equations (\ref{check1}) are automatically satisfied in quantum theory. Secondly, one will experimentally measure the success probability, i.e., the last one in \textbf{Theorem 1}, whose theoretical quantum prediction is given by
\begin{eqnarray}\label{check3}
&p(\hat{a}_1 \hat{a}_2 \hat{a}_3)= {\rm tr}[\rho (\hat{a}_1 \otimes\hat{a}_2\otimes \hat{a}_3)]=\frac{1}{4}.
\end{eqnarray}

\begin{figure}[t]
\includegraphics[width=55mm]{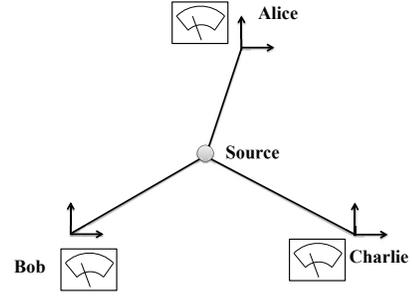}\\
\caption{(color online). Illustration of experimental observation for the stronger paradox in the $[n=3;|\alpha|=2,|\beta|=2]$-scenario.
   A three-qubit GHZ state is initially prepared in the source, from which each local qubit is then sent to Alice, Bob, and Charlie, respectively,
   who then randomly choose either of their two distinct directions ($\hat{x}$ and $\hat{y}$ directions) to perform the measurements. Within the experimental errors, one needs to verify the zero-probability condition (\ref{check1}) and measure the success probability in (\ref{check3}) test the quantum paradox.
 }\label{fig2}
\end{figure}

Taking into account experimental errors due to environment noise such that the six probabilities in (\ref{check1}) are not exactly zeros by measurements, let us denote the conditions as $p(\hat{b}_1 \hat{b}_2 \hat{a}_3)= p(\hat{b}_1 \hat{a}_2 \hat{b}_3)=p(\hat{a}_1 \hat{b}_2 \hat{b}_3)=p(\overline{\hat{b}}_1 \overline{\hat{b}}_2 \hat{a}_3)= p(\overline{\hat{b}}_1 \hat{a}_2 \overline{\hat{b}}_3)=p(\hat{a}_1 \overline{\hat{b}}_2 \hat{\overline{b}}_3)=\epsilon$. With the aid of the inequality $\mathcal{I}[3;2,2;1,1]= a_1 a_2 a_3-b_1 b_2 a_3-b_1 a_2 b_3-a_1 b_2 b_3-\overline{b}_1 \overline{b}_2 a_3-\overline{b}_1 a_2 \overline{b}_3-a_1 \overline{b}_2 \overline{b}_3 \le 0$, if one can observe the violation then he must have $1/4-6 \epsilon> 0$. Thus the maximal tolerant of measurement error is $\epsilon < 1/24\approx 0.041$.

\emph{Conclusions and discussion.}---While Hardy's paradox and Hardy's inequality have been generalized to arbitrary $n$ qubits by Cereceda, we have found that Cereceda's way of extension is not the unique one. In this paper, we have presented the most general framework for $n$-particle Hardy's paradox and Hardy's inequality. For $n\ge 3$ the generalized paradox may possess higher success probability, thus is stronger than the standard Hardy's paradox. And for GHZ states with $n\geq 3$, $|\alpha|=|\beta|= {\rm even\; number}$ is the best choice for generalized Hardy's paradoxes. For $n\ge 4$, the generalized Hardy's inequalities resist more noise than the standard Hardy's inequality (one can also adopt the generalized Hardy's inequality to prove Gisin's theorem, which we shall discuss elsewhere). Particularly in consideration of Table~\ref{vis} and~\cite{supp}, our result shows that for GHZ states with $n \geq 5$, the relation $|\alpha|=|\beta|=2$ is the best choice for generalized Hardy's inequality $\mathcal{I}[n;|\alpha|,|\beta|;x=1,y=1]\leq 0$. Moreover, in the three-qubit system, we have also designed  a feasible experiment proposal to observe the stronger quantum paradox. In our opinion, the results here advance the study of Bell's nonlocality both with and without inequality. we anticipate the experimental work in this direction in the near future.

\begin{acknowledgments}
S.H.J. and Z.P.X. contributed equally to this work. J.L.C. is supported by National Natural Science Foundations of China (Grant No.\ 11475089). H.Y.S. acknowledges the Visiting Scholar Program of Chern Institute of Mathematics, Nankai University.
\end{acknowledgments}


\begin{thebibliography}{99}
\bibitem{othertight}Numerical calculation has been done for $|\alpha|\leq n\leq 7$. Other tight inequalities are
$\mathcal{I}[n;n+1-k,k;1,1], (\forall n\leq 7, 2\leq k\leq n)$, $\mathcal{I}[n;n-2,2;2,1], (\forall n = 5, 6, 7)$, and
$\mathcal{I}[n=7;3,2;1,1]$.
\end{thebibliography}
\end{document}

\title{Supplementary Material for ``Generalized Hardy's Paradox''}

\author{Shu-Han Jiang}
\affiliation{Theoretical Physics Division, Chern Institute of Mathematics, Nankai University, Tianjin 300071, People's Republic of China}
\affiliation{School of Physics, Nankai University, Tianjin 300071, People's Republic of China}

\author{Zhen-Peng Xu}
\affiliation{Theoretical Physics Division, Chern Institute of Mathematics, Nankai University, Tianjin 300071, People's Republic of China}
\affiliation{Departamento de F\'{\i}sica Aplicada II, Universidad de  Sevilla, E-41012 Sevilla, Spain}

\author{Hong-Yi Su}
%\email{hysu@gscaep.ac.cn}
\affiliation{Graduate School of China Academy of Engineering Physics, Beijing 100193, People's Republic of China}

\author{Arun Kumar Pati}
%\email{akpati@hri.res.in}
 \affiliation{Quantum Information and Computation Group, Harish-Chandra Research Institute, Chhatnag Road, Jhunsi, Allahabad 211 019, India}

\author{Jing-Ling Chen}
%\email{chenjl@nankai.edu.cn}
\affiliation{Theoretical Physics Division, Chern Institute of Mathematics, Nankai University, Tianjin 300071, People's Republic of China}
\affiliation{Centre for Quantum Technologies, National University of Singapore, 3 Science Drive 2, Singapore 117543}

%\date{\today}

%\pacs{03.65.Ud, 03.67.Mn, 42.50.Xa}

\maketitle

\section{Proof of the optimum condition $|\alpha|=|\beta|$ for generalized Hardy's paradox}
If we apply the measurement strategy that we adopted in the proof of Theorem 2, then the success probability in the case of ``$|\beta| < |\alpha|$'' is given by
\begin{equation}
p(a_{I_n})\biggr|_{|\beta|<|\alpha|} = \frac{\gamma^2|e^{i\tau_0\pi}-\gamma^{n \tau_1 /2-1}|^2}{(1+\gamma^{2})
(1+\gamma^{\tau_1})^{n}},
\end{equation}
with $\tau_0=\frac{|\alpha||\beta|}{|\alpha|-|\beta|}$ and $\tau_1=\frac{2|\alpha| + 2|\beta|}{(n-|\beta|)|\alpha| + (n-|\alpha|)|\beta|}$. Since $\tau_0$ cannot be an odd number, one can have the following inequality
\begin{equation}
\label{S2}
p(a_{I_n}) \biggr|_{|\beta|<|\alpha|} < \frac{\gamma^2(1+\gamma^{n \tau_1 /2-1})^2}{(1+\gamma^{2})
(1+\gamma^{\tau_1})^{n}}.
\end{equation}

Let us compare this result (\ref{S2}) with the success probability in the case of ``$|\beta| = |\alpha| = {\rm even\; number}$'', which reads
\begin{equation}
p(a_{I_n}) \biggr|_{|\beta|=|\alpha|} = \frac{\gamma^2(1+\gamma^{n \tau'_1 /2-1})^2}{(1+\gamma^{2})
(1+\gamma^{\tau'_1})^{n}},
\end{equation}
with $\tau'_1=\tau_1 \biggr|_{|\beta|=|\alpha|} > \tau_1 \biggr|_{|\beta|<|\alpha|} $ for a given $|\alpha|$. Notice that $p_{a_{I_n}}$ is a monotonic increasing function of $\tau_1$, then one obtains
\begin{equation}
\label{S4}
p(a_{I_n}) \biggr|_{|\beta|<|\alpha|} < p(a_{I_n}) \biggr|_{|\beta|=|\alpha|}.
\end{equation}

The inequality (\ref{S4}) holds true for all $n \geq 3$ (for $n = 2$, the requirement ``$|\beta| = |\alpha| = {\rm even\; number}$'' contradicts with the prerequisite ``$|\alpha|+|\beta| \leq n+1$" in Theorem 1).
Furthermore, by direct calculation one can find that (i) for generalized GHZ states with $\gamma\neq 1$, ``$|\alpha| = |\beta| = 2 \lfloor \frac{n+1}{4} \rfloor$"  is the optimal choice, since in this case $\tau'_1$ takes its maximum value; (ii) for the standard GHZ state with $\gamma=1$, $|\alpha| = |\beta|= {\rm even\; number}$ is the optimal choice, because the success probability $p(a_{I_n})=1/{2^{n-1}}$ is bigger than those of the generalized GHZ states. Thus in conclusion, for the standard GHZ states with $n \geq 3$ the relation ``$|\alpha| = |\beta|= {\rm even\; number}$" is the optimum condition for generalized Hardy's paradox.

\section{Proof of the optimum condition $|\alpha|=|\beta|=2$ for generalized Hardy's inequality $\mathcal{I}[n;|\alpha|,|\beta|;x=1,y=1]\leq 0$}

Recall that for a given $n$-qubit entangled state $\rho$, one can mixed it with the white noise $I_{\rm noise}=\openone^{\otimes n}/2^n$, and then the resultant density matrix is given by $\rho^V =V \rho+ (1-V) I_{\rm noise}$. The threshold visibility $V_{\rm thr}$ is defined as the visibility below which Bell's inequality cannot be violated (i.e. $\tr [\rho^{V_{\rm thr}} \;\hat{\mathcal{I}} \;]=0$), where $\hat{\mathcal{I}}$ is the Bell operator corresponding to the generalized Hardy's inequality
 \begin{eqnarray}%\label{eq:genhardy2}
&\mathcal{I}[n;|\alpha|,|\beta|;x,y]= F(n;\alpha,\beta; x, y) \; p(a_{I_n})
- x \sum_{\alpha} p(b_{\alpha} a_{\bar{\alpha}}) - y \sum_{\beta } p(\bar{b}_{\beta} a_{\bar{\beta}}) \leq 0,
\end{eqnarray}
One can have the following expression
\begin{equation}
\label{S5}
V_{\rm thr} = \frac{\tr[\hat{\mathcal{I}}]}{\tr[\hat{\mathcal{I}}] - 2^n \; \mathcal{I}_{\rm QM}^{\max}[n;|\alpha|,|\beta|; x, y]},
\end{equation}
where $\tr[\hat{\mathcal{I}}] = F(n;|\alpha|,|\beta|; x, y) - x\binom{n}{|\alpha|} - y \binom{n}{|\beta|} < 0$, and $\mathcal{I}_{\rm QM}^{\max} \geq 0$ is the maximal quantum violation for the generalized Hardy's inequality. Eq. (\ref{S5}) implies that in order to get the lowest $V_{thr}$, one should appropriately select the values of $|\alpha|, |\beta|,x$ and $y$ to maximize $\mathcal{I}_{\rm QM}^{\max}$ and minimize $|\tr[\hat{\mathcal{I}}]|$ at the same time.

Based on the symmetry of the corresponding inequality and quantum state, the quantum violation can achieve its maximum value when (i) the measurement settings for the $n$ sites are the same; (ii) $\gamma = 1$, i.e. the corresponding quantum state is the $n$-qubit GHZ state. This is indeed one option for the extremum if the settings are symmetrical for each partite. One can refer to the first paper about multipartite Hardy paradox for details \cite{Cereceda2004}. By direct calculation, for $n$-qubit GHZ state, one has
\begin{subequations}
\begin{align}
&p(b_{\alpha}a_{\bar{\alpha}}) = \frac{1}{2^{n}} \biggr( 1 + \cos \left( \theta_1 + |\alpha| \theta_2 \right) \biggr),\\
&p(\bar{b}_{\beta}a_{\bar{\beta}}) = \frac{1}{2^{n}} \biggr( 1 + \cos \left( \theta_1 + |\beta| \theta_2  + |\beta| \pi \right) \biggr),\\
&p(a_{I_{n}}) = \frac{1}{2^{n}} \biggr( 1 + \cos  \theta_1  \biggr),
\end{align}
\end{subequations}
with $\theta_1 = n \theta_a,~\theta_2 = \theta_b - \theta_a$. By substituting these probabilities into the generalized Hardy's inequality, we get
\begin{eqnarray}
\mathcal{I}_{\rm QM}^{\max}[n;|\alpha|,|\beta|;x,y] &=& \frac{1}{2^n} \; \max_{\theta_1,\theta_2} \biggr \{ F(n;|\alpha|,|\beta|; x,y) \left( 1 + \cos \theta_1  \right)-\nonumber\\
&&  x \binom{n}{|\alpha|} \left( 1 + \cos \left( \theta_1 + |\alpha| \theta_2 \right) \right) - y \binom{n}{|\beta|} \left( 1 + \cos \left( \theta_1 + |\beta| \theta_2  + |\beta| \pi \right) \right)\biggr \} \geq 0,
\end{eqnarray}
with
\begin{equation}
F(n;\alpha,\beta; x, y) =\min_{0\le m \le n} \left(x\binom{m}{|\alpha|} + y \binom{n-m}{|\beta|}\right).
\end{equation}
Finally we can have
\begin{equation}
\label{S8}
\begin{split}
V_{\rm thr} &= \min_{\theta_1,\theta_2} \biggr \{\frac{-F(n;|\alpha|,|\beta|; x, y) + x\binom{n}{|\alpha|} + y \binom{n}{|\beta|}}{F(n;|\alpha|,|\beta|; x, y) \cos \theta_1  - x\binom{n}{|\alpha|} \cos \left( \theta_1 + |\alpha| \theta_2 \right) - y \binom{n}{|\beta|}\cos \left( \theta_1 + |\beta| \theta_2  + |\beta| \pi \right)}\biggr \}\\
& \geq \frac{-F(n;|\alpha|,|\beta|; x, y) + x\binom{n}{|\alpha|} + y \binom{n}{|\beta|}}{F(n;|\alpha|,|\beta|; x, y) + x\binom{n}{|\alpha|} + y \binom{n}{|\beta|}},
\end{split}
\end{equation}
in which the equality holds if and only if $|\alpha| = |\beta| = {\rm even\; number}$.

For simplicity, we may recast (\ref{S8}) in the following form:
\begin{equation}
V_{\rm thr} \geq 1 - \frac{2}{1 + V'},
\end{equation}
with
\begin{equation}
V'[n;|\alpha|,|\beta|;x,y] = \frac{x \binom{n}{|\alpha|} + y \binom{n}{|\beta|}}{F(n;|\alpha|,|\beta|; x, y)}.
\end{equation}
Thus the problem of the minimization of $V_{thr}$ is then equivalent to the minimization of $V'$. However, the problem is quite complicated for the case of ``$x \not= y$'', because one cannot exhaust all the possibilities even by using numerical computations and there is no obvious regular pattern. Therefore, here we only focus on the case of ``$x = y$''. Actually the case ``$x = y$'' is equivalent to the case ``$x = y=1$'', because one can divide the generalized Hardy's inequality $\mathcal{I}[n;|\alpha|,|\beta|; x, y]\leq 0$ by $x$ and then can obtain an equivalent inequality $\mathcal{I}[n;|\alpha|,|\beta|; x=1, y=1]\leq 0$. In this particular case, one has
\begin{equation}
V'[n;|\alpha|,|\beta|;x,x] = V'[n;|\alpha|,|\beta|;1,1] = \frac{\binom{n}{|\alpha|} + \binom{n}{|\beta|}}{F(n;|\alpha|,|\beta|;1,1)}
\end{equation}
Generally, it is difficult to obtain the analytical expression for the function $F(n;|\alpha|,|\beta|;1,1)$ for arbitrary $|\alpha|$ and $|\beta|$. However, for $|\beta|=|\alpha|$, one can have the analytical expression as \begin{equation}
  F[n;|\alpha|,|\alpha|;1,1] =  \binom{\lfloor (n+1)/2 \rfloor}{|\alpha|} + \binom{n - \lfloor (n+1)/2 \rfloor}{|\alpha|}.
\end{equation}
Based on which, one can have the visibility as
$V_{\rm thr}[n;q,q;1,1]=\frac{2\binom{n}{q}-\binom{[n/2]}{q}-\binom{n-[n/2]}{q}}{2\binom{n}{q}+\binom{[n/2]}{q}+\binom{n-[n/2]}{q}}$ with $|\alpha|=|\beta|=q$. For $q=2$, we have $V_{\rm thr}[n;2,2;1,1]= \frac{3[\frac{n+1}{2}]-1}{5[\frac{n+1}{2}]-3}$.

In the following steps, we shall prove our main conclusion: For the case of $x=y=1$  (or $x=y$), for GHZ states with $n \geq 5$, the relation $|\alpha|=|\beta|=2$ is the best choice for generalized Hardy's inequality.

\emph{Step 1.---} For small particle number $n$, because for a fixed $n$ one can exhaust all the possibilities by varying $|\alpha|$ and $\beta$ in the numerical computation, hence one can check the conclusion directly by numerical calculation for the cases of $n=3, 4, 5, 6$. For $n = 3$ and $4$, the best choice is $|\alpha| = 3, ~|\beta| = 1$, which yields $V_{\rm thr}[n=3;3,1;1,1] = \sqrt{-3 + 2 \sqrt{3}} \le V_{\rm thr}[n=3;2,2;1,1]$  and $V_{\rm thr}[n=4;3,1;1,1]=\frac{9}{7}\sqrt{\frac{3}{11}}\le V_{\rm thr}[n=4;2,2;1,1]$, respectively. However, For $n = 5$ and $6$, the best choice is just $|\alpha|=|\beta|=2$.

\emph{Step 2.---} For $n \geq 7$, we shall prove the above conclusion systematically. We need to prove the following result
\begin{equation}
V'[n;|\alpha|,|\beta|;x,x] \geq V'[n;2,2;x,x] = \frac{4\lfloor\frac{n+1}{2}\rfloor-2}{\lfloor\frac{n+1}{2}\rfloor-1}.
\end{equation}
The corresponding proof is divided into two cases:

\emph{Case (i):}  $|\alpha| \geq 3, |\beta|\geq 2$. From the knowledge of Combinatorics, we know that
\begin{equation}
\min_{0\leq m \leq n} \left\{ \binom{m}{|\alpha|} + \binom{n-m}{|\beta|} \right\} \leq 2^{-n} \sum_{m=0}^n \binom{n}{m} \left\{ \binom{m}{|\alpha|} + \binom{n-m}{|\beta|} \right\} =
2^{-|\alpha|}\binom{n}{|\alpha|} + 2^{-|\beta|}\binom{n}{|\beta|}.
\end{equation}
Thus,
\begin{equation}
V'[n;|\alpha|,|\beta|;x,x] \geq \frac{ \binom{n}{|\alpha|} + \binom{n}{|\beta|} }{2^{-|\alpha|}\binom{n}{|\alpha|} + 2^{-|\beta|}\binom{n}{|\beta|}} := W[n;|\alpha|,|\beta|].
\end{equation}
Note that $W[n;|\alpha|,|\beta|] \geq 2^{|\alpha|} \geq 8$ (thus greater than $V'[n;2,2;x,x]$ strictly) when $|\alpha| \geq |\beta|\geq 3$.

For $|\beta|=2, n-2 \geq |\alpha|\geq 3$, we come to prove that $W[n;|\alpha|,|\beta|] > W[n;n-1,2]$ firstly. It is equivalent to prove
\begin{equation}
2^{n-3}\binom{n}{|\alpha|}\binom{n}{2} > (2^{n-3}-1) n\binom{n}{2} + 2^{n-|\alpha|-1}\binom{n}{|\alpha|}\binom{n}{2} + (2^{n-|\alpha|-1}-1) n\binom{n}{|\alpha|}.
\end{equation}
This is indeed true by noting that when $n\geq 7$,
\begin{eqnarray}
\frac{1}{3} \cdot 2^{n-3}\binom{n}{|\alpha|}\binom{n}{2} &\geq& 2^{n-3}\frac{(n-1)}{6}n\binom{n}{2}> (2^{n-3}-1) n\binom{n}{2},\\
\frac{1}{6} \cdot 2^{n-3}\binom{n}{|\alpha|}\binom{n}{2} &\geq& 2^{n-3-1}\frac{(n-1)}{6}n\binom{n}{|\alpha|} > (2^{n-|\alpha|-1}-1) n\binom{n}{|\alpha|},\\
\frac{1}{2} \cdot 2^{n-3}\binom{n}{|\alpha|}\binom{n}{2} &\geq& 2^{n-|\alpha|-1}\binom{n}{|\alpha|}\binom{n}{2}.
\end{eqnarray}
Then we move to prove that $W[n;n-1,2] > \frac{4(n-1)}{n-2} \geq V'[n;2,2;x,x] $. This can be deduced to prove
\begin{equation}
2^{n-4}(n-3) \geq (n-1),
\end{equation}
which is an obvious fact for $n\geq 7$.

\emph{Case(ii):}   $|\beta|=1$, $|\alpha| > |\beta|$. In this case, we have
\begin{equation}
\label{S19}
F(n;|\alpha|,1; x, x) = \min_{0\leq m \leq n} \left\{ \binom{m}{|\alpha|} + \binom{n-m}{1} \right\} = n - |\alpha| + 1.
\end{equation}
Therefore we can get
\begin{equation}
  V'[n;|\alpha|,1;x,x] = \frac{ \binom{n}{|\alpha|} + n }{(n - |\alpha| + 1)} := U[n;|\alpha|].
\end{equation}

Notice that for $n\geq 8$, we have $U[n;2] = \frac{n(n+1)}{2(n-1)} > \frac{4(n-1)}{n-2}\geq V'[n;2,2;x,x]$.  And for $n = 7$,  we have $U[7;2] = V'[7;2,2;1,1]$, but in this case  $V_{\rm thr}$ is strictly greater than $1-\frac{2}{1+U[7;2]}$.
Now we come to prove that $U[n;|\alpha|] \geq U[n;2]$ for $n\geq 7$, which can be deduced to prove that
\begin{equation}
\label{S20}
2(n-1)\binom{n}{|\alpha|} \geq n(n+1)^2 - |\alpha|n(n+1)-2n(n-1).
\end{equation}
Direct calculation shows that $U[n;|\alpha|] \geq U[n;2]$ when $|\alpha| = n, n-1$. For $|\alpha| \leq n-2$. The relation (\ref{S20}) is true by noting that
\begin{equation}
n(n+1)^2 - |\alpha|n(n+1)-2n(n-1) \leq n(n+1)^2 - 2n(n+1)-2n(n-1) = 2(n-1)\binom{n}{2}\leq 2(n-1)\binom{n}{|\alpha|} .
\end{equation}
Hence for $n \geq 7$ the minimum value of $V'$ (i.e. the minimum of $V_{\rm thr}$) is achieved at $|\alpha|=|\beta|=2$.

Thus, in conclusion, for the case of $x=y=1$  (or $x=y$), for GHZ states with $n \geq 5$, the relation $|\alpha|=|\beta|=2$ is the best choice for generalized Hardy's inequality.